\documentclass[11pt]{article}

\usepackage{amsmath}
\usepackage{amsfonts}
\usepackage{latexsym}
\usepackage{amssymb}
\usepackage{graphicx}
\usepackage{amsthm}
\usepackage{color}
\usepackage{chemarrow}
\usepackage{fullpage}
\usepackage{framed}
\usepackage{verbatim}
\usepackage{geometry}
\usepackage{color}
\usepackage{epsfig}
\usepackage{paralist}
\usepackage{hyperref}
\usepackage{makeidx}
\usepackage{authblk}

\usepackage{tikz}

\DeclareMathOperator{\pay}{pay}
\DeclareMathOperator*{\argmax}{\arg\!\sup}
\DeclareMathOperator*{\E}{\mathbb{E}}

\theoremstyle{definition}
\newtheorem{theorem}{\bf{Theorem}}[section]

\newtheorem{remark}[theorem]{Remark}
\newtheorem{definition}{\bf{Definition}}[section]

\theoremstyle{remark}
\newtheorem{example}{\bf{Example}}[section]

\newtheorem{notation}{Notation}[definition]

\newcommand{\MR}{\mathbb{R}}
\newcommand{\MZ}{\mathbb{Z}}
\newcommand{\MN}{\mathbb{N}}

\newcommand{\op}{\operatorname}

        {\begin{list}
                {\noindent\makebox[0mm][r]{$\bullet$}}
                {\leftmargin=5.5ex \usecounter{enumi}
 		 \topsep=1.5mm \itemsep=-.75ex}
        }
        {\end{list}}

\title{Playing games in an uncertain world}
\author{\small Manoj Gopalkrishnan\thanks{manojg@tifr.res.in}, Girish Varma\thanks{girishrv@tifr.res.in}}
\affil{School of Technology and Computer Science,\\ Tata Institute of Fundamental Research, Mumbai, India.}

\begin{document}
\maketitle
\begin{abstract}
Traditional game theory assumes that the players in the game are aware of the rules of the game. However, in practice, often the players are unaware or have only partial knowledge about the game they are playing. They may also have knowledge that other players have only partial knowledge of the game they are playing, which they can try to exploit. We present a novel mathematical formulation of such games. We make use of Kripke semantics, which are a way to keep track of what different players know and do not know about the world. We propose a notion of equilibrium for such games, and show that equilibrium always exists.
\end{abstract}


\section{Introduction}

Recall the Monty Hall problem. A games show host shows a contestant three doors, and tells the contestant that there are prizes behind one of the doors, and no prize behind the other two. The contestant is asked to pick a door. The games show host opens one of the other doors, revealing that it has no prize. He asks the contestant if they want to switch their choice, or stick with their original choice.

It is well known that most people intuitively feel it makes no difference if they switch or not, whereas the best strategy is to indeed switch. The standard explanation given for this apparently suboptimal human behavior has been that human beings are not good with Bayes' theorem, and probabilistic reasoning. This may well be true. However, there is another possible explanation. Perhaps the suboptimal behavior has to do with the contestant not being completely clear on what the rules of the game are. 

Perhaps the contestant is not sure about the following questions: Did the show host know that the door he threw open did not have anything behind it? Did he use that fact while making a choice of that door? Or did he randomly select a door, and then reveal it to have no prize behind it? What is the incentive for the games show host: help the contestant win a prize, make sure the contestant wins as less as possible, or entertain the audience? 

Ambiguity in the rules of the game being played is very common in real life. The payoff matrix for a game often depends on information that the players may possess only partially. To complicate things, players may have partial information about about the information available to other players. This can lead to higher-order strategic reasoning based on what other players know and do not know. 

Our contributions are:
\begin{itemize}
\item Kripke semantics is a well-developed field concerned with reasoning about knowledge. By combining Kripke semantics with game theory, we come up with a mathematical formulation we call \textbf{Kripke games} (Definition~\ref{def:kripke-game}). Our formulation allows a precise, direct, and granular treatment of the effect of epistemic considerations on strategy.
\item We propose a notion of \textbf{equilibrium} for Kripke games (Definition~\ref{def:kripkeeqm}). Our notion assumes that every player supposes the worst about her missing knowledge, and that this pessimism among players is common knowledge. This definition, within its limitations, has the merit that it succinctly captures the effect of arbitrary levels of epistemic reasoning among the players.
We prove in Theorem~\ref{thm:kripkeeqm}  that every Kripke game has an equilibrium.
\item We introduce the notion of \textbf{min-games} (Definition~\ref{def:mingame}). The idea is that the players are simultaneously involved in multiple games, and each player's payoff is the minimum across all these games. In Theorem~\ref{thm:kripkemin}, we prove that every Kripke game with pessimistic players reduces to a min-game, and in Theorem~\ref{thm:mingame}, we prove that every min-game has an equilibrium. 
\end{itemize}
Apart from their mathematical role in the above theorems, min-games appear to be interesting objects of study in their own right. One-player min-games correspond to optimization tasks where the payoff is the minimum among various payoffs. This setting is closely related to multi-objective optimization~\cite{miettinen1999nonlinear}, which is an active research area. Thus min-games may be viewed as a natural extension of multi-objective optimization to the context of games.

\section{An Example}\label{sec:example}

In traditional two-player zero-sum games~\cite{von2007theory}, one considers a game between two players, {\em Row} and {\em Column}. Each player has a choice of two pure strategies, strategy $1$ and strategy $2$. Their payoffs are described by a single $2\times 2$ \textbf{payoff matrix} with rows indexed by strategies of {\em Row}, and columns indexed by strategies of {\em Column}. {\em Row} tries to maximize payoff, and {\em Column} tries to minimize it.

What if the players don't know exactly what game they are playing? For concreteness, consider three possible worlds, call them world $1$, world $2$, and world $3$, with three different payoff matrices, as follows:
\begin{itemize}
\item In world $1$, the 
payoff matrix is given by $\left(\begin{array}{cc}\
-1 & 1\
\\1 & -1\
\end{array}\right)$. 
\item In world $2$, the 
payoff matrix is given by
$\left(\begin{array}{cc}
2 & 0\
\\0 & 1\
\end{array}\right)$. 
\item In world $3$, the payoff matrix is given by
$\left(\begin{array}{cc}
1 & -1\
\\0 & 2\
\end{array}\right)$.
\end{itemize}

If both players know the ``true world'' (or a probability distribution over the possible worlds) then this reduces to a two-player zero-sum game in the traditional sense of von Neumann and Morgenstern, and can be solved in the usual way. However, suppose the players don't even know a probability distribution over the possible worlds. In addition, though the players don't know which world is the true world, perhaps {\em Row} knows that the true world is either $1$ or $2$, whereas {\em Column} knows that the true world is either $2$ or $3$, and {\em Row} has some idea about what {\em Column} may know, and vice versa. This sort of strategic thinking is what we are interested in exploring. We will capture this through the device of \textbf{Kripke semantics}~\cite{kripke1959completeness} on the set of worlds.

We're going to introduce two equivalence relations $\equiv_r$ and $\equiv_c$, corresponding to the row and column player respectively, on the set $\{1,2,3\}$ of worlds. The idea is that all $\equiv_r$ worlds look identical to \emph{Row}, and similarly all $\equiv_c$ worlds look identical to \emph{Column}. For concreteness, we define $\equiv_r$ and $\equiv_c$ as the smallest equivalence relations containing the relations $1\equiv_r 2$ and $2\equiv_c 3$.  In detail:
\begin{itemize}
\item The equivalence relations themselves are assumed to be common knowledge between {\em Row} and {\em Column}. So both {\em Row} and {\em Column} know $1\equiv_r 2$ and $2\equiv_c 3$, and both know that both know it, and so on. 
\item If the true world is world $2$, then {\em Row} knows that the true world is either world $1$ or world $2$, since $1\equiv_r 2$, whereas {\em Column} knows that the true world is either world $2$ or world $3$, since $2\equiv_c 3$. 
\item If the true world is world $1$, then {\em Row} knows that the true world is either $1$ or $2$, since $1\equiv_r 2$. {\em Column} knows that the true world is world $1$ since $w\equiv_c 1$ implies $w = 1$. Further, {\em Column} knows that {\em Row} knows that the true world is either $1$ or $2$. And so on. 
\end{itemize}

So if the true world is world $3$, {\em Row} knows that the true world is world $3$. So {\em Row} could simply play the optimal strategy for Game $3$. However, {\em Row} knows that {\em Column} does not know if the true world is world $2$ or world $3$! Possibly {\em Row} can exploit this knowledge to do even better.

A computation shows that under the notion of equilibrium introduced in Definition~\ref{def:kripkeeqm}, the optimal strategy for \emph{Row} is to play either row with probability $1/2$ if the world is either $1$ or $2$, and to play the second row if the world is $3$. The optimal strategy for \emph{Column} is to play either column with probability $1/2$ if the world is $1$. If the world is $2$ or $3$, then \emph{Column} plays the first column with probability $0.6$. The corresponding expected payoffs in the three worlds are $0, 0, 4/5$ to \emph{Row}, and $0, -4/5, -4/5$ to \emph{Column} respectively.

\section{Kripke Games}
We first recall the well-known notions of finite (non-zero sum, non-cooperative) games, taking the opportunity to establish notation.

\begin{definition}[Finite Game]\label{def:game}
Let $P$ be a finite set. A \textbf{finite game with players $P$} is described by: for each player $p\in P$, a finite set $S_p$ of \emph{strategies}, and a \emph{payoff} $u_p:\prod_{q \in P} S_q \rightarrow \MR$.
\end{definition}

\begin{definition}[Mixed strategy]\label{def:mixstr}
Fix a finite game with players $P$. 
\begin{enumerate}
\item For each player $p\in P$, a \textbf{mixed strategy} is a probability distribution $\sigma_p$ over $S_p$. 
\item A \textbf{play} is a probability distribution $\sigma = \prod_{p\in P} \sigma_p$ on $\prod_{p\in P} S_p$, where each $\sigma_p$ is a mixed strategy for player $p$. 
\item A \textbf{play instance} is a random variable $X=(X_p)_{p\in P}$ such that $\op{law}(X)$ is a play, i.e., each $X_p$ takes values in $S_p$, and $\{X_p\}_{p\in P}$ are all independent.
\item The \textbf{payoff} $\pay_p(\sigma)$ to player $p$ from a play $\sigma$ is the expected payoff $\E[u_p(X)]$ where $X$ is a play instance with law $\sigma$.
\end{enumerate}
\end{definition}

%
%


We are now ready to introduce Kripke games.

\begin{definition}[Kripke Game]
\label{def:kripke-game}
Let $P$ and $W$ be finite sets. A \textbf{Kripke game with players $P$ on worlds $W$} is described by: for each player $p\in P$,
\begin{itemize}
\item A finite set $S_p$ of \emph{strategies},
\item An equivalence relation $\equiv_p$ on the set $W$ of worlds,
\item For each world $w\in W$, a \emph{payoff} $u^w_p:\prod_{q\in P} S_q\to \MR$.
\end{itemize}
\end{definition}

\begin{remark}\label{rmk:union}
Note that the set $S_p$ of strategies of player $p$ is the same in all worlds. There is no loss of generality in this assumption since, if a player is allowed different strategies on different worlds, we can always reduce to this case by the following trick. We take the union of all possible strategies and declare that to be the set $S_p$. We extend the payoff matrix, by making the payoff $-A$ if the player plays a previously non-existent strategy in a world, where $A$ is a real number sufficiently large to ensure that the player will never want to play that strategy.
\end{remark}

\begin{notation}
$[W]_p$ will denote the set of $\equiv_p$-equivalence classes of $W$. If $w\in W$ then $[w]_p\in [W]_p$ will denote the $\equiv_p$-equivalence class containing $w$. That is, $[w]_p = \{v\in W\mid v\equiv_p w\}$.
\end{notation}

\begin{definition}[Strategy]\label{def:st-kripke-game}
Fix a Kripke game with players $P$ on worlds $W$.
\begin{enumerate}
\item A \textbf{pure strategy} $s_p$ for player $p$ is a map from $[W]_p$ to $S_p$.
\item A \textbf{mixed strategy} $\sigma_p$ for player $p$ is a probability distribution over the pure strategies $S_p^{[W]_p}$ of player $p$.
\item A \textbf{play} is a probability distribution $\sigma = \prod_{p\in P} \sigma_p$, where each $\sigma_p$ is a mixed strategy.
\item A \textbf{play instance} is a random variable $X=(X_p)_{p\in P}$ such that $\op{law}(X)$ is a play, i.e., each $X_p$ takes values in $S_p^{[W]_p}$, and $\{X_p\}_{p\in P}$ are all independent.

\item The \textbf{specialization} $X(w)$ of a play instance $X=(X_p)_{p\in P}$ to the world $w$ is the random variable $(X_p([w]_p))_{p\in P}$ taking values in $\prod_{p\in P} S_p$. We will view $X(w)$ as a play instance on the finite game in the world $w$.

\item The \textbf{payoff} $\pay^w_p(\sigma)$ for player $p$ in world $w$ from play $\sigma$ is the minimum over all worlds $v\in [w]_p$ of the expected payoff $\E[u^v_p(X(v))]$, where $X$ is a play instance with law $\sigma$. That is:
\[
	\pay^w_p(\sigma) := \min_{v\in [w]_p} \E[u^v_p(X(v))].
\]
\end{enumerate}
\end{definition}

The definition of the payoff $\pay^w_p$ requires further explanation. Note that if $v\equiv_p w$ then in fact $\pay^v_p = \pay^w_p$. The idea is that if the true world is $w$, then player $p$ only knows that the true world belongs to the equivalence class $[w]_p$. So player $p$ takes a worst-case attitude, assuming that her payoff will be the worst payoff over all worlds $v\in [w]_p$. 

\begin{notation}
Fix a Kripke game with players $P$. We will denote by $\Sigma_p$ the set of mixed strategies for player $p$, and by $\Sigma$ the set $\prod_{p\in P} \Sigma_p$ of all plays. If $\sigma=\prod_{p\in P} \sigma_p$ is a play, and $\sigma'_q$ is a mixed strategy for player $q$ then $\sigma - \sigma_q + \sigma'_q $ will denote the play $\sigma'_q\times \prod_{p\in P\setminus\{q\}} \sigma_p $ obtained from $\sigma$ by replacing the mixed strategy $\sigma_q$ of player $q$ by the mixed strategy $\sigma'_q$.
\end{notation}

\begin{definition}[Equilibrium in a Kripke Game]\label{def:kripkeeqm}
Fix a Kripke game with players $P$ on worlds $W$.
\begin{enumerate}
\item A player $p\in P$ \textbf{tolerates} a play $\sigma$ iff for all worlds $w\in W$, for all mixed strategies $\sigma'_p\in\Sigma_p$, the payoff $\pay^w_p(\sigma)\geq \pay^w_p(\sigma - \sigma_p + \sigma'_p)$.


\item A play $\sigma$ is an \textbf{equilibrium} iff every player tolerates $\sigma$.
\end{enumerate}
\end{definition}

Note that our notion of equilibrium assumes that each player is pessimistic/ risk-averse, \emph{and that this is common knowledge}. What is the psychological validity of our assumption? There is some evidence in the psychology literature in the work of Kahneman and Tversky~\cite{kahneman1979prospect} that human beings are indeed risk-averse in decision making. It is an important question, and possibly one that lends itself to empirical investigation, whether and in what settings human beings treat such risk-averseness among all human players as common knowledge. 

We emphasize that Kripke games should admit various other notions of equilibrium, for example, inspired by ideas from bounded rationality, or repeated games, which we will not explore further in this paper. Our main theorem is the following.

\begin{theorem}\label{thm:kripkeeqm}
Every Kripke game has an equilibrium play.
\end{theorem}

In Section~\ref{sec:mingames}, we define \textbf{min-games}, and show that they permit a notion of equilibrium. In Section~\ref{sec:kripketomin}, we show that Kripke games reduce to min-games, so that equilibrium plays in min-games allow us to obtain equilibrium plays in Kripke games.


\section{Min-Games}\label{sec:mingames}
\begin{example}
Consider two zero-sum finite games between \emph{Row} and \emph{Column}, described by the payoff matrices $\left(\begin{array}{cc}\
-1 & 2\
\\1 & -3\
\end{array}\right)$ and 
$\left(\begin{array}{cc}
2 & -4\
\\0 & 5\
\end{array}\right)$, where \emph{Row} wants to maximize payoff, and \emph{Column} wants to minimize. Suppose \emph{Column} plays the first column. If \emph{Row} plays the first row, her payoff is $-1$ from the first game, and $2$ from the second game. If she plays the second row, her payoff is $1$ from the first game, and $0$ from the second game. 

In a min-game, we assume that the payoff to a player is the worst-case of her payoff from all the games. So \emph{Row}'s min-game payoff in response to \emph{Column} playing the first column is $-1$ if she plays the first row, and $0$ if she plays the second row. Her best response to \emph{Column} playing the first column is thus to play the second row. 

\emph{Column}'s payoffs to the play (row $2$, column $1$) are $1$ and $0$. Since \emph{Column} is trying to minimize, her worst-case payoff is the maximum of these two values, i.e., $1$. Note that in this play, \emph{Row}'s min-game payoff is coming from the second game, and \emph{Column}'s min-game payoff is coming from the first game. So it is possible for different players to get their min-game payoff values from different games. We now formalize these ideas, and introduce a notion of equilibrium for min-games.
\end{example}
%

\begin{definition}[Min-Game]\label{def:mingame}
Fix a finite set $P$ and a positive integer $k\in\MZ_{\geq 1}$. A \textbf{min-$k$-game} with players $P$ is described by: for each player $p\in P$, a set $S_p$ of \emph{strategies}, and \emph{payoffs} ${u^1_p,u^2_p,\dots,u^k_p : \prod_{q\in P} S_q \to \MR}$.
\end{definition}

\begin{remark}
The intuitive idea is that there are $k$ parallel games going on. We are assuming that the set of strategies for each player is the same in all the $k$ games she is playing. This causes no loss of generality, by an argument parallel to the argument offered in Remark~\ref{rmk:union}.
\end{remark}

The definitions of mixed strategy, play, and play instance carry over to min-games unchanged from the setting of finite games in Definition~\ref{def:mixstr}. The key difference is in the definition of payoff.

\begin{definition}\label{def:mingamepay}
Fix a min-$k$-game with players $P$. 
\begin{enumerate}
\item The \textbf{payoff} $\pay_p(\sigma)$ to player $p$ from a play $\sigma$ in a min-$k$-game is defined by
\[
\pay_p(\sigma) := \min\left\{ \E[u^1_p(X)],\E[u^2_p(X)],\dots,\E[u^k_p(X)]\right\}
\] 
where $X$ is a play instance with law $\sigma$. 
\item Player $p$ \textbf{tolerates} play $\sigma$ iff for all mixed strategies $\sigma_p'$ for player $p$, the payoff \[{\pay_p(\sigma - \sigma_p + \sigma_p')\leq \pay_p(\sigma)}.\]
\item\label{def:mingameeq} An \textbf{equilibrium} is a play tolerated by every player.
\end{enumerate}
\end{definition}

We now state the main theorem about min-games, which shows existence of equilibria.

\begin{theorem}\label{thm:mingame}
Every min-game has an equilibrium play.
\end{theorem}

\begin{remark}
Consider a one-player min-game where the player has two strategies $s_1$ and $s_2$, and the payoff functions are $u^1(s_1) = 1, u^1(s_2) = 2, u^2(s_1) = 3, u^2(s_2) = 0$. Then $\pay(p s_1 + (1-p) s_2)$ is the minimum of $p + 2(1-p)$ and $3p$. This minimum is achieved at $p^* = 1/2$ and has value $3/2$. 

In particular, note that this game can not be represented by a one-player finite game, because every one-player finite game has a linear payoff function, whereas for this game the payoff function is not linear.
\end{remark}

We will give two proofs for  Theorem \ref{thm:mingame}. The first proof will employ Kakutani's fixed point theorem in a more general way than Nash's proof for existence of equilibrium in finite games. Recall the statement of Kakutani's fixed point theorem from \cite{osborne1994course}.

\begin{theorem}[Kakutani's fixed point theorem]\label{thm:kfp}
Let $S$ be a non-empty, compact, and convex subset of some Euclidean space $\MR^n$. Let $B:S\to 2^S$ be a set-valued function on $S$ with a closed graph and the property that $B(x)$ is non-empty and convex for all $x\in S$. Then $B$ has a fixed point.
\end{theorem}

\subsection{The first proof}

\begin{proof}[I for Theorem~\ref{thm:mingame}]
Fix a min-$k$-game with players $P$. For each player $p\in P$, let $\Sigma_p$ denote the set of mixed strategies of player $p$. Let $\Sigma = \prod_{p\in P}\Sigma_p$ be the set of all plays. We will define an \emph{improvement map} $B:\Sigma\to\Sigma$. The map represents what happens if all the players unilaterally and simultaneously improve their mixed strategy in response to the play $\sigma$. This map $B$ will satisfy the requirements of Kakutani's fixed point theorem, allowing us to obtain our result.

We first define the map $B_p:\Sigma\to 2^{\Sigma_p}$ for player $p\in P$ by \[
\sigma\in\Sigma\longmapsto\argmax_{\sigma'_p \in \Sigma_p} \pay_p(\sigma - \sigma_p +\sigma'_p).
\]
It is a set-valued map because the range is all the $\sigma'_p\in\Sigma_p$ at which the supremum is attained. Now define $B:\Sigma\to 2^\Sigma$ as the map that takes $\sigma\in\Sigma$ to the product set $\prod_{p\in P} B_p(\sigma)$.

We claim that a fixed point $\sigma$ of $B$ must be an equilibrium in the sense of Definition~\ref{def:mingamepay}.\ref{def:mingameeq}. This is immediate by checking that each player tolerates $\sigma$, which follows directly from the way $B$ was defined, and from the definition of fixed point. We now show that the conditions of Kakutani's fixed point theorem (Theorem~\ref{thm:kfp}) are met by the map $B$.
\begin{itemize}
\item Since $\Sigma_p$ are non-empty simplexes,  the set $\Sigma= \prod_{p \in P} \Sigma_p$ of plays is convex, compact and non-empty.
\item $B(\sigma)$ is non-empty for all $\sigma\in \Sigma$. This is because for each player $p\in P$, the set $B_p(\sigma)$ is non-empty, since $\pay_p$ is continuous, and the supremum of a continuous function over a compact set is always attained.
%
\item We claim that $B(\sigma)$ is convex for all $\sigma\in \Sigma$. It is enough to argue that each set $B_p(\sigma)$ is convex since the finite product of convex sets is convex. The key observation is that, fixing $\sigma$ and $p$, the function $f(\sigma'_p):= \pay_p(\sigma - \sigma_p + \sigma'_p)$ is the minimum of a finite number of linear functions, and is hence a (piecewise-linear) concave function. In other words, it is straightforward to verify that:
\begin{align*}
\pay_p(\sigma - \sigma_p + (\sigma^1_p+ \sigma^2_p)/2) 
&\geq \frac 1 2 \left( \pay_p(\sigma - \sigma_p + \sigma^1_p) +  \pay_p(\sigma - \sigma_p +  \sigma^2_p) \right)
\end{align*}
where $\sigma^1_p,\sigma^2_p\in\Sigma_p$. If $\sigma^1_p,\sigma^2_p\in B_p(\sigma)$, then $(\sigma^1_p+ \sigma^2_p)/2\in B_p(\sigma)$ because
\[\pay_p(\sigma - \sigma_p + \sigma^1_p) = \pay_p(\sigma - \sigma_p +  \sigma^2_p) = \sup_{\sigma'_p \in \Sigma_p} \pay_p(\sigma - \sigma_p +\sigma'_p),\]
and our claim is proved.
%
%
%
\item We claim that the graph $\{ (\sigma,B(\sigma)) \mid \sigma\in \Sigma\}$ of $B$ is closed. This will essentially follow from the continuity of the payoff function. Consider a sequence $(\sigma(n),\tau(n))_{n \in \MN}$ such that $\tau(n) \in B(\sigma(n))$ for all $n \in \MN$. Suppose this sequence converges to $(\sigma^*,\tau^*)$. Then continuity of $\pay$ implies that $\pay(\tau(n))\to\pay(\tau^*)$. 

Suppose for contradiction that $\tau^*\notin B(\sigma^*)$. Then there exists a player $p\in P$ with a mixed strategy $\tau'_p\in\Sigma_p$ such that:
\[
\pay_p(\tau^* -\tau^*_{p}+ \tau'_p) > \pay_p(\tau^*) = \lim_{n\to\infty} \pay_p(\tau(n)).
\]
We conclude that $\pay_p(\tau(n) - \tau_p(n) + \tau'_p) > \pay_p(\tau(n))$ for large $n$, because $\pay_p(\tau(n) - \tau_p(n) + \tau'_p)$ can be made arbitrarily close to $\pay_p(\tau^* -\tau^*_{p}+ \tau'_p)$ by continuity of $\pay_p$, and $\pay_p(\tau(n))$ can be made arbitrarily close to $\pay_p(\tau^*)$. This is a contradiction, since we had supposed that $\tau(n)\in B(\sigma(n))$. Our claim is proved.
\end{itemize}
This completes the proof.
\end{proof}

\subsection{The second proof}
\begin{proof}[II for Theorem~\ref{thm:mingame}]
Fix a min-$k$-game with players $P$. We will construct a finite game with  players $P \cup \hat{P}$ where $|\hat{P}| = |P|$ and $P\cap\hat{P}=\emptyset$, with the property that each Nash equilibrium in this finite game corresponds to an equilibrium for the min-game. Since Nash equilibria exist for finite games, the theorem follows.

We now describe the finite game. For each player $p\in P$, introduce a new and distinct player $\hat{p}\in\hat{P}$. We call $\hat{p}$ the \textbf{world-chooser opponent} of $p$. The set $S_{\hat{p}}$ of strategies of $\hat{p}$ equals the set $W$ of worlds.

We now describe the payoffs. Fix a player $p\in P$. Then the payoff $u_p : \prod_{q\in P\cup\hat{P}} S_q \to\MR$ will
 depend only on $S_{\hat{p}} \times \prod_{q\in P} S_q$, and ignore the strategies played by the world-chooser 
 opponents of the other players. If $w\in S_{\hat{p}}=W$ and $s\in  \prod_{q\in P} S_q$ then we define 
 $u_p(w,s):= u_p^w(s)$. The payoff $u_{\hat{p}}$ also depends only on the strategy played by the players $P$, and the 
 player $\hat{p}$, and is given by $u_{\hat{p}}(w,s) = -u_p(w,s) = - u_p^w(s)$.

In words, the payoff to player $p$ comes from the payoff function in the world chosen by $\hat{p}$. The world-chooser opponent $\hat{p}$ gets a payoff that is the negative of the payoff of $p$. The effect of this is that player $\hat{p}$ wants to minimize $p$'s payoff, by choosing the world in which $p$'s payoff is minimum. It remains to verify that a Nash equilibrium in this finite game corresponds to our notion of equilibrium for a min-game. 

Consider a play $\sigma$. Define the \textbf{lowered} play $\ell(\sigma):=\prod_{p\in P} \sigma_p$. Then  
\begin{align}\label{eqn:minnash}
\pay_p(\ell(\sigma)) \leq \pay_p(\sigma)
\end{align} for each player $p\in P$. Here $\pay_p(\ell(\sigma))$ and $\pay_p(\sigma)$ denote the payoff to player $p$ from the \emph{min-game play} and the \emph{finite game play} respectively. Equality holds for all plays $\sigma$ tolerated by $\hat{p}$, in particular for all Nash equilibria.

If $\sigma$ is a Nash equilibrium then we claim that $\ell(\sigma)$ is an equilibrium for the original min-game. For suppose player $p$ is considering another mixed strategy $\sigma'_p$. We have:
\[
\pay_p(\ell(\sigma) - \sigma_p + \sigma'_p)=\pay_p(\ell(\sigma - \sigma_p + \sigma'_p)) \leq \pay_p(\sigma - \sigma_p + \sigma'_p) \leq \pay_p(\sigma) = \pay_p(\ell(\sigma))
\]
The first inequality is by Equation~\ref{eqn:minnash}. The last equality is by Equation~\ref{eqn:minnash}, using the fact that $\sigma$ is a Nash equilibrium. Hence every player tolerates $\ell(\sigma)$, and so it is an equilibrium by definition.
\end{proof}

\section{Reducing Kripke games to min-games}\label{sec:kripketomin}
\begin{theorem}\label{thm:kripkemin}
Every Kripke game has an equilibrium play.
\end{theorem}
\begin{proof}
We will deduce this theorem by reducing a Kripke game to a min-game with the property that the min-game's equilibria correspond to equilibria of the original Kripke game. The result then follows from Theorem \ref{thm:mingame}. 

Fix a Kripke game with players $P$ on worlds $W$. We will construct a min-$|W|$-game with players \[\mathcal{P}=\cup_{p\in P} [W]_p.\] That is, each player of the original Kripke game is split into multiple players, one for each $\equiv_p$-equivalence class. Player $[w]_p$ is assigned the strategy set $S_p$ of player $p$ in the original Kripke game. On world $v\in W$, the payoff 
\[
\mathfrak{u}^v_{[w]_p}\left(\prod_{\mathfrak{p}\in\mathcal{P}} s_\mathfrak{p} \right) 
= 
\begin{cases}
u_p^v\left(\displaystyle\prod_{q\in P} s_{[w]_q} \right) 
\text{ if } v\equiv_p w\
\\+A \text{ otherwise.}
\end{cases}
\]
where $u_p^v$ is the payoff from the original Kripke game, and $A>0$ is a positive constant chosen sufficiently large so that every player's minimum payoff is guaranteed to be less than $A$, in every play. 

In words, we have ensured that the payoff to player $[w]_p$ comes only from worlds $v\equiv_p w$. The only opponents who matter to $[w]_p$ are players $[w]_q$ where $q\in P$.

Given a play $\sigma$ for this min-game, we construct a play $\tau$ for the original Kripke game as follows: the mixed strategy $\tau_p^{[w]_p}$ is defined as $\sigma_{[w]_p}$. The key property is that
\[
	\pay^w_p(\tau) = \pay_{[w]_p}(\sigma)
\]
where the LHS is the payoff for the Kripke game, and the RHS for the min-game. This allows us to prove that if $\sigma$ is an equilibrium play for this min-game then $\tau$ is an equilibrium play for the original Kripke game. If not, then player $p$ can improve her payoff in world $w$ in the Kripke game, which will translate to an improved payoff for player $[w]_p$ in the min-game, contradicting the assumption that $\sigma$ was an equilibrium play.
\end{proof}

%

\section{Related work}
Games with incomplete information, and epistemic notions in game theory, have appeared before, in particular in several works by Aumann. The closest to our work appears to be Aumann's work on Interactive Epistemology~\cite{aumann1999interactive}. There are parallels between Aumann's \emph{semantic approach} and our approach. Aumann's space $\Omega$ of states of the world appears to correspond to our set $W$ of worlds. His partitions on the world appear to correspond to the equivalence class partition induced by the equivalence relations $\equiv_p$ for each player $p$. The focus in \cite{aumann1999interactive} appears to be to provide epistemic justifications for many of the standard assumptions of game theory, whereas our focus is on how epistemic considerations affect strategy.
\bibliographystyle{amsplain}
\bibliography{KripkeGames}

\providecommand{\bysame}{\leavevmode\hbox to3em{\hrulefill}\thinspace}
\providecommand{\MR}{\relax\ifhmode\unskip\space\fi MR }
\providecommand{\MRhref}[2]{%
  \href{http://www.ams.org/mathscinet-getitem?mr=#1}{#2}
}
\providecommand{\href}[2]{#2}
\begin{thebibliography}{1}

\bibitem{aumann1999interactive}
Robert~J. Aumann, \emph{Interactive epistemology i: Knowledge}, International
  Journal of Game Theory (1999).

\bibitem{kahneman1979prospect}
Daniel Kahneman and Amos Tversky, \emph{Prospect theory: An analysis of
  decision under risk}, Econometrica \textbf{47} (1979), no.~2, 263--91.

\bibitem{kripke1959completeness}
Saul~A. Kripke, \emph{A completeness theorem in modal logic}, The Journal of
  Symbolic Logic \textbf{24} (1959), no.~1, pp. 1--14 (English).

\bibitem{miettinen1999nonlinear}
K.~Miettinen, \emph{Nonlinear multiobjective optimization}, Kluwer Academic
  Publishers, Boston, 1999.

\bibitem{von2007theory}
John~Von Neumann and Oskar Morgenstern, \emph{Theory of games and economic
  behavior (commemorative edition)}, Princeton university press, 2007.

\bibitem{osborne1994course}
M.J. Osborne and A.~Rubinstein, \emph{A course in game theory}, MIT Press,
  1994.

\end{thebibliography}
\end{document}